\documentclass{fundam}
\usepackage{amsmath}
\usepackage{theorem}
\usepackage{proof}
\usepackage{graphicx}
\usepackage{xspace}
\usepackage{color}
\usepackage[all]{xy}
\usepackage{mathtools}
\usepackage{tikz-cd}

\usepackage{macros}

\begin{document}
\author{Miko{\l}aj Boja\'nczyk}
\title{Separator logic and star-free expressions for graphs}
\maketitle
\begin{abstract}
    We describe two formalisms for defining graph languages, and prove that they are equivalent: 
    \begin{enumerate}
        \item \emph{Separator  logic.} (Introduced independently in by Schrader, Siebertz and Vigny in~\cite{schrader2021firstorder}.) This is first-order logic on graphs which, apart from the binary relation ``vertices $x$ and $y$ are connected by an edge'', has for every  $n \in \set{0,1,\ldots}$ a relation of arity $n+2$ which says that  ``vertex $x$ can be connected to vertex $y$ by a path that avoids vertices $z_1,\ldots,z_n$''.
        \item \emph{Star-free graph expressions.} These are expressions that  describe graphs with distinguished vertices called ports, and  which are  built from finite languages via Boolean combinations and  operations on graphs with ports used to construct tree decompositions.
    \end{enumerate}
    Furthermore, we prove a variant of \schutz's theorem (about star-free languages being those recognized by aperiodic monoids) for  graphs of bounded pathwidth. A corollary is that, given $k$ and  a graph language represented by an \mso formula, one can decide if the language can be defined in either of two above (equivalent) formalisms on graphs of pathwidth at most $k$.
\end{abstract}

\section{Introduction}
In this introduction, we argue that although first-order definable languages of graphs, words and trees are intensively studied, there is a mismatch between the considered variants of first-order logic. For graphs, a local version that uses the edge relation only is typically studied. On the other han,d  for words and trees one typically studies a variant that has order and not just successor. A consequence of this mismatch is that if we view a word as a special case of a graph, then the first-order definable word languages of words will not translate to first-order definable languages of graphs. To overcome this mismatch, we propose a stronger variant of first-order logic for graphs, which can be seen as the natural graph generalisation of first-order logic with an order on positions. 

Let us describe this situation in more detail, starting with first-order logic  for words and trees. 
In the field of logic and automata, when  first-order logic is used to define a property  of words, the vocabulary usually contains the order relation $x \le y$, and not just the successor relation $x + 1 =y$.  The difference is unimportant for  monadic second-order logic \mso, but important for first-order logic, since order  cannot be defined using first-order logic in terms of successor (successor can be defined in terms of order).   Word languages that can be defined  first-order logic with order have been intensively studied, starting in the 1960's,  and are known to have many equivalent descriptions, which are summarized in the following diagram
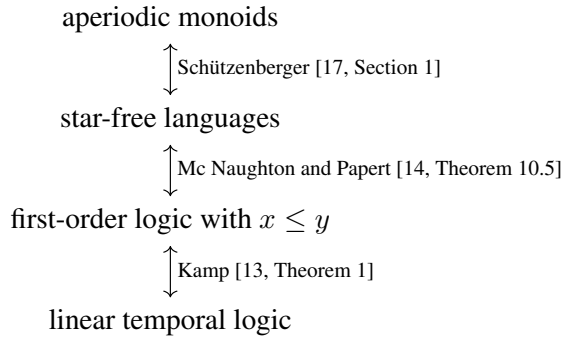
\[
\begin{tikzcd}
 \txt{aperiodic monoids}
 \ar[d,<->,"\text{\schutz~\cite[Section 1]{Schutzenberger65}}"]
 \\ 
    \txt{star-free languages}
    \ar[d,<->,"\text{Mc Naughton and Papert~\cite[Theorem 10.5]{McNaughtonPapert71}}"]
\\
\txt{first-order logic with $x \le y$}
\ar[d,<->,"\text{Kamp~\cite[Theorem 1]{Kamp68}}"]
\\
\txt{linear temporal logic} 
\end{tikzcd}
\]
Also for trees, first-order logic is usually studied together with the descendant ordering, and not just successor (we use the name ``child'' instead of successor  when talking about trees).  Here, the appropriate temporal logic is CTL*~\cite[Main Theorem]{haferThomas1987}, and there is also a corresponding notion of  star-free languages~\cite[Section 4]{bojanczykForestExpressions2007}. This gives tree counterparts for three of the four descriptions for word languages in the diagram above. (The missing counterpart is ``aperiodic monoids''. Finding an algebraic characterization of first-order definable tree languages remains an major open problem~\cite[Problem 3]{bojanczyk2015automata}.)  Generally speaking, in this line of research the focus is on understanding the expressive power of the logic, ideally by giving an algorithm which inputs a regular language and decides if it can be defined in the logic. Logics for words and trees that have only the successor relation have also been studied,  even if they might seem slightly less fundamental than the variants with order, and there exist  algebraic and decidable characterizations, see~\cite[p.252]{BrzozowskiSimon73} for the word case and ~\cite[Theorem 1]{benediktSegoufin2009} for the tree case.

A different attitude is prevalent for graphs. Here,  the usual notion of first-order logic uses the edge relation only, and has no predicates for reachability. As mentioned before, this difference is unimportant for monadic second-order logic. In first-order logic, however, reachability cannot be defined in terms of the edge relation. In the study of first-order logic on graphs, the focus is  on  finding efficient algorithms for model checking, with a famous result being that every sentence of first-order logic can be evaluated in almost linear time on every class of graphs that is nowhere dense~\cite[Theorem 1.1]{groheKreuzterSiebertz2017ACM}. From a technical point of view, reasoning about first-order logic on graphs  (with the edge relation only) usually relies on Gaifman locality, while reasoning about first-order logic on words and trees (with order) usually  relies on compositionality.

As we can see from the above discussion, the traditional study of first-order logic has considered  three cells in the following table:
\begin{center}
    \begin{tabular}{r|cc}
         & neighbor & reachability \\
         \hline
         words and trees &  discussed & discussed \\
         graphs & discussed & \\
    \end{tabular}
\end{center}
The purpose of this paper is to fill the missing cell, by proposing a variant of first-oder logic on graphs with reachability, which we call separator logic. We show that the class of languages definable in this logic is reasonably robust, by presenting an equivalent notion of star-free expressions. The notions are designed so that if we view words and trees as a special case of graphs, then we recover the previously studied classes of first-order definable languages with order. Finally, we show an algebraic characterization, in the style of \schutz's aperiodic monoids, of languages definable in separator logic for graphs of bounded pathwidth. An extension of this characterization from bounded  pathwidth to bounded treewidth seems to be beyond the reach of current methods, since it would require an algebraic characterization of first-order logic with descendant on trees.

\paragraph*{Acknowledgements.} I would like to thank Nicole Schrader, Sebastian Siebertz and Alexandre Vigny  for interesting discussions, which happened after we learned that we had been independently studying the same logic. Also, I would also like to thank Colin Geniet, Micha{\l} Pilipczuk and Szymon Toru\'nczyk for their helpful remarks. Finally, I would like to acknowledge the financial support of the ERC Consolidator Grant \emph{Lipa}, grant agreement 683080.

\section{Separator logic and star-free languages}
In this section, we describe the two formalisms used in this paper, namely separator logic and star-free languages of graphs, and we show that they are equivalent. Graphs are finite and undirected.  

\subsection{Separator logic}
\label{sec:separator-logic}
We begin our discussion with the logic, which is based on adding an infinite family of relations, apart from the edge relation, so that one can talk about separators in a graph. More formally,  define its \emph{separator model} of a graph  to be the relational structure where the universe is the vertices, and which is equipped with the following relations:
\begin{align*}
\myunderbrace{E(x,y)}{there is an edge from\\
\scriptsize  vertex $x$ to vertex $y$} 
\hspace{3cm} 
\myunderbrace{S_n(x,y,z_1,\ldots,z_n)}{for every $n \in \set{0,1,\ldots}$ there is\\
\scriptsize a relation $S_n$ of arity $n+2$ which \\
\scriptsize says that every path from $x$ to $y$ must  \\
\scriptsize use some vertex from $\set{z_1,\ldots,z_n}$}.
\end{align*}
The relations $S_n$ are called \emph{separator relations}.
There are infinitely  many separator relations in the model, and they have  unbounded arity, but of course every formula uses  finitely many relations. By convention, the relation $S_0(x,y)$ says that vertices $x$ and $y$ are in different connected components of the graph.
We use the name \emph{separator logic} for first-order logic using the separator model. 

\begin{example}\label{ex:connected-and-cycle}
Example, the following formula says that a graph is disconnected:
\begin{align*}
\myunderbrace{\exists x \ \exists y\ S_0(x,y)}{there is a separator of size $0$\\
\scriptsize between some two vertices}.
\end{align*}
Here is another sentence, which says that the graph contains a cycle:
\begin{align*}
\myunderbrace{\exists x \ \exists y\  \forall z\ \neg S_1(x,y,z) }{
    there are two vertices that cannot \\
    \scriptsize be separated by one vertex
}.
\end{align*}
By saying that a graph is connected and has no cycles, we can define trees in separator logic.
\end{example}
In this paper, we are most interested in the expressive power of separator logic. An alternative research direction would be to search for efficient algorithms for model checking of  separator logic on restricted graph classes; this alternative is pursued in~\cite{schrader2021firstorder}. We do not study that research direction, beyond the following example, which was proposed by Micha{\l} Pilipczuk and Szymon Toru\'nczyk.

\begin{example}[Complexity of model checking] This example concerns the model checking problem for  nowhere dense graph classes; in this example we assume that the reader is familiar with nowhere dense graph classes and the complexity of their model checking problem. For first-order logic with the edge relation only, the model checking problem (with the parameter being the formula) is fixed parameter tractable for every nowhere dense class,  see~\cite[Theorem 1.1]{groheKreuzterSiebertz2017ACM}. In contrast,  the same model checking problem is not fixed parameter tractable for the class of all graphs, or more generally any graph class that is closed under subgraphs and not nowhere dense, subject to a standard assumption in the field of parameterized complexity, namely that  AW[*] is different from FPT.
    
    In this example, we show a graph class  that is nowhere dense (in fact, it has the stronger property of bounded expansion),  such that  model checking of separator logic for this graph class is as hard as model checking of the usual first-order logic (with the edge relation only) over the class of all graphs. Therefore,  the model checking problem for separator logic is unlikely to be  fixed parameter tractable over this class.  In other words, bounded expansion alone is not sufficient for tractability of  the model checking problem for separator logic. 

    The idea is to simulate edges using long   paths, thus making a graph sparse,  and yet keeping enough information about the  original graph so that it can be recovered using separator logic. For a graph, define its \emph{subdivision} to be the result of subdividing each edge with  $n$ fresh vertices, where $n$ is the number of vertices in the original graph, as explained in a the following picture:
    \mypic{17}
    For simplicity of presentation, the subdivision is seen as vertex coloured graph, with the vertices of the original graph being black, and the subdividing vertices being red. To handle vertex coloured graphs, separator logic is extended with unary predicates for testing if a vertex has a given colour.   The vertex colours could be eliminated by using suitable gadgets. 
    Define 
    $\Cc$ to be the subdivisions -- in the sense described above -- of all graphs. 
     We argue below that $\Cc$ has bounded expansion, and model checking separator logic on $\Cc$ is at least as 
     hard as model checking first-order logic with the edge relation only on the class of all graphs. 
    
     \begin{enumerate}
         \item We first prove that $\Cc$  has bounded expansion. We assume that the reader is familiar with $r$-minors and bounded expansion.  For $r \in \set{0,1,\ldots}$, define $\Cc_r$ to be the class of depth-$r$ minors of graphs from $\Cc$. To prove that $\Cc$ has bounded expansion, we need to show that for every choice of $r$, graphs from $\Cc_r$ have a number of edges that is at most linear in the number of vertices. To see this, one observes that with finitely many exceptions,  every graph from $\Cc_r$ has the property that every edge is adjacent to at least one vertex of degree at most two, and the latter property implies that the number of edges is at most twice the number of vertices.
         \item We observe that a graph can be recovered from its subdivision using separator logic. The vertices  are the black vertices, and the edge relation corresponds to ``can be connected by a path that uses only red vertices''. To check if two black vertices $v$ and $w$ are connected by a red path, we check if removing $v$ and $w$ gives a graph that has a nonempty purely red connected component; this can be expressed in separator logic. 
     \end{enumerate}
     The hardness result from this above example uses only  the separator predicate $S_2$. This is optimal, since first-order  logic with the edge relation and  the separator predicates $S_0$ and $S_1$  is no harder than first-order logic with edges only, as far as the model checking problem is concerned.   For the  predicate $S_0$ alone this is easy to see,  since it is enough to compute the first-order theories of the connected components of a graph. A similar, if more complicated, argument can be made for $S_1$, by computing the first-order theories in the tree of 2-connected components. 
    \end{example}

\begin{example}[Words as graphs]
    In this example, we show that if words are viewed as graphs, then separator logic has the same expressive power as the usual notion of first-order logic for words that has the position order. Like in the previous example, to simplify notation  we use vertex coloured graphs.  For a word $w \in \Sigma^*$, define the \emph{path graph of $w$} to be the vertex coloured graph, with the colours being $\Sigma$ plus an extra black colour, that is  described in the following picture: 
    \mypic{18}
    The special black vertex on the left is used to orient the word. One could potentially avoid it  by using directed edges; however it is far from clear what is the right notion of separator logic for directed graphs. The left-to-right ordering on vertices of the path graph (with the special black vertex being the leftmost one) can be easily defined in separator logic: a non-black vertex $u$ is to the left of a non-black vertex $w$ if removing $u$ separates $w$ from the black vertex.  Conversely, the separator predicates in a path graph can be defined in first-order logic based on the order. This proves that a word language $L \subseteq \Sigma^*$ is definable in first-order logic with the position order if and only if the graph language 
    \begin{align*}
    \set{ \text{path graph of $w$} : w \in L}
    \end{align*}
    is definable in separator logic.  A similar argument also works for trees. Note that for path graphs, and also for graphs that encode trees, we do not need all separator predicates, only $S_1$. 
\end{example}

\subsection{Star-free expressions for graphs}
\label{sec:star-free-for-graphs}
We now move to the second formalism considered in this paper, which is star-free expressions. These will be proved to have  the same expressive power as separator logic. 

The star-free expressions are  based on graphs with ports\footnote{Graphs with ports, and the operations on them, are based on Courcelle~\cite[Section 1.7]{courcelle1987}. Courcelle uses the word ``source'' instead of ``port''. We use the name ``treewidth operations'' because a class of graphs has  treewidth $k$ if and only if it can be generated using the operations starting  from graphs with at most edge, so that at most $k+1$ ports are used at any given moment. }.
Define a \emph{graph with $i$ ports} to be a graph with a tuple of $i$ distinguished vertices, called \emph{ports}. All ports must be pairwise different. Here is a picture of a graph with two ports:
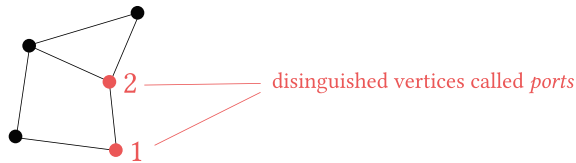
\mypic{1}
Since graphs with ports are the only kind of graphs that use, we call them graphs from now on. The \emph{arity} of a graph is defined to be the number of ports.
We use the following operations on graphs with ports.

\begin{definition}[Treewidth operations]\label{def:operations-on-graphs}
    The set of  treewidth operations is the following (infinite) set of operations. 
    
    \begin{itemize}
        \item {\bf Fusion.} For every $k \in \set{0,1,\ldots}$ there is a  \emph{$k$-fusion} operation which  inputs two graphs of arity $k$, and outputs the graph of arity $k$ that is obtained by taking the disjoint union of the two input graphs, and then identifying, for every $i \in \set{1,\ldots,k}$, the $i$-th port of both input graphs, as explained in the following picture:
        \mypic{2}
        The two inputs to $k$-fusion might disagree on the subgraphs induced by the ports. As a result of fusion, the edges are accumulated: if ports $v$ and $w$ are connected by an edge in at least one of two input graphs, then they are connected by an edge in the output graph. 
        \item {\bf Forget.} For every $k \in \set{0,1,\ldots}$ there is a \emph{$k$-forget} operation which inputs a graph of arity $k+1$, and outputs the graph of arity $k$  that is obtained from the input graph by no longer distinguishing port $k+1$.
        \mypic{3}
        \item {\bf Add. } For every $k \in \set{0,1,\ldots}$ there is a \emph{$k$-add} operation that inputs a graph of arity $k$, and outputs the graph of arity $k+1$ that is obtained from the input graph by adding an extra isolated vertex, which becomes port $k+1$.
        \mypic{4}
                \item {\bf Permute.}  For every $k \in \set{0,1,\ldots}$ and every  permutation of $\set{1,\ldots,k}$, there is a unary operation on graphs of arity $k$, which reorders the list of ports according to the permutation.
    \end{itemize}

\end{definition}

Graphs together with the treewidth operations can be viewed as a multisorted algebra, with  the sorts being  $\set{0,1,\ldots}$ and corresponding to arities. 
Define a \emph{graph language} to be a set of graphs, all of which have the same arity. The arity of the language is defined to be the arity of some (equivalently, every) graph in the language. When we take the  complement of a graph language, we mean the complement with respect to all graphs of the given arity. (A corner case is the empty graph language: there is an empty language for every arity $k$. The complement of this graph language is  the set of all graphs of arity $k$.)
The treewidth operations can  be applied to graph languages in the natural way, e.g.~the $k$-fusion of two languages of arity $k$ consists of all $k$-fusions where the first input is from the first language and the second input is from the second language. 
We now define the main object of this note, which is star-free languages of graphs. 

\begin{definition}[Star-free graph language]\label{def:star-free}
    The class of star-free graph languages is the least class of graph languages that contains all finite languages, and is closed under Boolean operations (including complementation) and the treewidth operations from Definition~\ref{def:operations-on-graphs}\footnote{An alternative presentation of the same idea would involve monads. We could take the hypergraph monad from~\cite[Section 3]{bojanczykTwoMonadsGraphs2018}, which is a monad that corresponds to treewidth, and define a star-free language to be any language that can be constructed from finite ones using Boolean combinations and term operations where each variable appears at most once. This alternative presentation indicates that there is a notion of star-free language in every monad, at least as long as one can express what it means for a variable to occur at most once in a term. One could study this notion in the vertex replacement monad for clique-width~\cite[Section 5]{bojanczykTwoMonadsGraphs2018}, such a study is mentioned as future work at the end of this paper.}.
\end{definition}

In the above definition, the induction basis for the constructors is the finite languages. An alternative would be to use only the empty sets and the  singleton sets for graphs with at most two vertices, since the remaining graphs can be constructed from such graphs using the treewidth operations.
We finish this section with some examples of star-free graph languages.
\begin{example}[Connected graphs]
    We describe a star-free expression that defines the set of graphs of arity zero which are connected. We assume that, by definition, every graph is nonempty, i.e.~it has at least one vertex (if we would allow the empty graph, then non-emptiness could be defined by taking a graph with one port, and forgetting that port). The set of all graphs of arity zero is therefore the complement of the empty set:
       \begin{align*}
        \neg \myunderbrace{(\emptyset : 0)}{the empty language of arity 0}.
       \end{align*}
    A graph of arity zero is disconnected if and only if it can be decomposed as the fusion of two nonempty graphs of arity zero.  Therefore, the set of connected graphs of arity zero is defined by the  following star-free expression (which uses $\oplus$ for fusion):
    \begin{align*}
    \neg \myunderbrace{(\neg (\emptyset : 0) \oplus \neg (\emptyset : 0))}{disconnected graphs}.
    \end{align*}
\end{example}
\begin{example} [Contains an given induced subgraph] \label{ex:induced-subgraph} Consider the language $L$ of graphs of arity $k$ which contain 
    an edge from port $i$ to port $j$. This language is star-free, because  it can be described as the fusion of the set of all graphs of arity $k$, with the singleton language that contains exactly one graph -- the graph of arity $k$ where all vertices are ports and there is only one edge, namely from port $i$ to port $j$. By taking a boolean combination of such graphs, we can use a star-free expression to specify which edges between ports are present, and which edges are not present. By forgetting all ports, we can define the language of graphs of arity zero which contain  some fixed graph $H$ as an induced subgraph. 
\end{example}

The above example demonstrates how the ports can be used to simulate variables of logic, and how forgetting can be used to simulate existential quantification. This simple idea will be used to prove that all languages definable in separator logic are star-free. 

\begin{example}[Cycles]
    A graph is a cycle if and only if it is connected, and every vertex has exactly two neighbors. This can be defined by a star-free expression.
\end{example}

\begin{example}[Trees] \label{ex:trees} To define the graphs which are trees, we use the same approach as in  Example~\ref{ex:connected-and-cycle}: a graph is a tree if and only if it is connected, and every two non-neighboring vertices can be separated by some other vertex. Connectivity was already treated in the examples above. For the separation property, we will  describe  a language that describes the separator predicate $S_1(x_1,x_2,x_3)$, namely the set   of graphs of arity 3 where port 3 separates ports 1 and 2. A graph belongs to  if and only if it can be obtained by taking the fusion  of two graphs of arity 3, such that in the first graph  port 2 is an isolated vertex, and in the second graph port 1 is an isolated vertex, as explained in the following picture:
    \mypic{21}
     If we forget port 3 in the graph language described above, then we get the language $K$ graphs of arity 2 where the two ports are separated by some vertex. 
     A connected graph of arity zero is a tree if for every way of selecting two non-neighboring ports, the resulting graph of arity 2 belongs to $K$.
\end{example}

\subsection{Equivalence of the two formalisms}
Having defined separator logic and star-free expressions, 
in this section we prove that they define the same graph languages. To define a graph language of arity $k \in \set{0,1,\ldots}$ in separator logic, we use formulas with free variables: for a formula $\varphi(x_1,\ldots,x_k)$ of separator logic with $k$ free variables, its language is defined to be the set of graphs of arity $k$ which satisfy the formula under the valuation which sets the $i$-th free variable to the $i$-th port. (This definition assumes an implicit order on the free variables.) Since the ports in a graph are required to be distinct, the graph language corresponding to a formula will only consider valuations where all free variables are different. For this reason, the language of the formula $x_1 = x_2$ is empty. We could have also considered a variant of graphs with ports which allow equalities on ports, with the same results, but we choose to work with the assumption that all ports are distinct.

\begin{theorem}\label{thm:star-free-is-logic}
    A graph language is star-free if and only if it is  definable in separator logic.
\end{theorem}

The two implications in the theorem are proved in Sections~\ref{sec:from-logic-to-expressions} and~\ref{sec:from-expressions-to-logic} below.

    \subsubsection{From separator logic to star-free expressions}
    \label{sec:from-logic-to-expressions}
    We begin with the easier implication, which is  from separator logic to star-free  expressions. Since the ports in star-free expressions can be used to represent free variables, we can use a simple induction on formula size: we show that  for every  formula with free variables contained in $\set{x_1,\ldots,x_k}$, the corresponding  language of arity $k$ is star-free. (The formula does not need to use all variables. This  happens for instance when taking a disjunction of two formulas that talk about different subsets of the variables.)  For boolean combinations there is nothing to do, since star-free languages have boolean combinations built in. Consider now an existential quantifier 
    \begin{align}
        \label{eq:existential-formula}
     \exists x_{k+1} \ \varphi(x_1,\ldots,x_k,x_{k+1}).
    \end{align}
    The rough idea is that existential quantification corresponds to forgetting the last port. There  is one slightly subtle point here: in the definition of the language of a formula, we only consider valuations where all variables represent distinct vertices, since the definition of a graph with ports requires all ports to be distinct.   Therefore, the language of the formula~\eqref{eq:existential-formula} consists of (a) the language of the formula $\varphi$ with the forget operator applied to it; plus (b) the language of the formula 
    \begin{align*}
           \bigvee_{i \in \set{1,\ldots,k}} \varphi(x_1,\ldots,x_k,x_i).
    \end{align*}
    The languages used in (a) and (b) are star-free thanks to the induction assumption, thus proving the induction step.  
   
    We are left with the induction basis, which corresponds to the edge and  separator predicates in the logic. Consider first the edge relation. The corresponding language is the graphs of arity $k$ where some two ports $i,j \in \set{1,\ldots,k}$ are connected by an edge; this language is star-free as we have shown in Example~\ref{ex:induced-subgraph}.  Consider now the the separator predicate. The corresponding graph language consists of graphs of arity $k$, such that some two ports $s,t \in \set{1,\ldots,k}$ are separated by a subset of ports $I \subseteq \set{1,\ldots,k}$. We need to justify that this language is star-free. Here, we use a similar idea as in Example~\ref{ex:trees}. More formally, it  is not hard to see that a graph belongs to this language if and only if 
         \begin{itemize}
             \item[(*)] there exists a partition of $\set{1,\ldots,k}-I$ into disjoint sets $I_1,\ldots,I_\ell$   such that the graph can be decomposed as a $k$-fusion
             \begin{align*}
             G_1 \oplus \cdots \oplus G_\ell
             \end{align*}
             so that $s$ and $t$ are in different blocks of the partition, and for every $i \in \set{1,\ldots,\ell}$ all of the ports from outside $I \cup I_i$ are isolated in  $G_i$. 
         \end{itemize}
     The idea behind (*) is that two ports are in the same block of the partition if they are equal or they can be connected by a path that avoids ports with indices in $I$. Finally, condition (*) is easily seen to be definable, because the partition can be chosen in finitely many ways, and because  ``port $i$ is an isolated vertex'' is a star-free language.

     \subsubsection{From star-free expressions to separator logic}
     \label{sec:from-expressions-to-logic}
     We now turn to the harder implication from Theorem~\ref{thm:star-free-is-logic}, namely that for every star-free expression, its language can be defined in separator logic. The proof is by induction  on the size of the expression.  The interesting case is the $k$-fusion  operation. By induction assumption, we know that both fused languages  can be defined in separator logic.

    The difficulty is that the definition of fusion cannot be directly formalized in separator logic, since this would require saying that one can  partition the non-port vertices into two parts which induce graphs in the two fused languages.  On the face of it, this would require quantifying over sets of vertices, which is a feature that is not available in  separator logic.
    However, we can work around this difficulty by using  a  compositionality argument.
    The key idea for this proof is that the type of a graph (which is  the information about the graph with respect to separator logic of given quantifier rank) can be inferred from the types of its prime factors, which are connected components of the graphs after the ports have been removed. As we will also see, the types of the prime factors will also give us enough information to decide if a graph belongs to the fusion of two languages definable in separator logic. A more detailed argument is presented below. 

        We use the usual notion of quantifier rank. For a quantifier rank $r \in  \set{0,1,\ldots}$, we  say that two  graphs  are $r$-equivalent if they have the  same arity $k \in \set{0,1\ldots}$  and they belong to the same languages defined by formulas of separator logic that have  quantifier rank at most $r$. 

    \begin{claim}[Congruence]\label{claim:fusion-congruence}
        Let $ r, k \in \set{0,1,\ldots}$, and let $\equiv$  be $r$-equivalence on graphs of arity $k$. Then $\equiv$  is a congruence with respect to $k$-fusion, i.e.
        \begin{align*}
        \bigwedge_{i = 1,2} G_i \equiv G'_i \quad \Rightarrow \quad G_1 \oplus G_2  \equiv G'_1 \oplus G'_2
        \end{align*}
holds for all graphs $G_1,G_2,G'_1,G'_2$ of arity $k$.
        \end{claim}
    \begin{proof}
        A standard \ef pebble game argument. We assume that the reader is familiar with such arguments, and hence we only sketch it. For $i \in \set{1,2}$,  define the  \emph{$i$-th local game} to be the \ef  game corresponding to the equivalence of $G_i$ and $G'_i$, and define the  \emph{composite game} to be the game corresponding to the equivalence of $G_1 \oplus G_2$ and $G'_1 \oplus G'_2$. Using a standard composition of strategies, Duplicator uses a strategy in the composite game  which has the property that  whenever a position in the composite game is reached, then for every $i \in \set{1,2}$, if we keep only the vertices and pebbles from $G_i$ and $G'_i$, then we get a position in the $i$-th local game that is consistent with Duplicator's winning strategy. We will now prove that this composite strategy is winning  in the composite game. To prove this, we need to show that if all $r$ rounds have been played, thus reaching two pebblings
        \begin{align}\label{eq:pebbling}
            \myunderbrace{x_1,\ldots,x_r}{a pebbling in $G_1 \oplus G_2$} 
            \qquad\text{and}  \qquad
            \myunderbrace{x'_1,\ldots,x'_r,}{a pebbling in $G'_1 \oplus G'_2$} 
        \end{align}
        then the same quantifier-free formulas are satisfied on both sides. The interesting case is that of a quantifier-free formula which is  a separator relation. To show that the same separator relations are satisfied by the two pebblings in~\eqref{eq:pebbling}, for every set $I \subseteq \set{1,\ldots,r}$ and $i, j \in \set{1,\ldots,r}$ we need to show that  if in one of the pebblings   the $i$-th pebble can be connected with the $j$-th pebble by a path that avoids pebbles from  $I$, then the same is true  in the other pebbling. Suppose that such a connecting path exists in one of the pebblings, say the left one.  To show that a similar connecting path exists in the right pebbling, we decompose the path on the left side  into segments, such that in each segment the ports are not used except for the source and target vertices. By definition of $k$-fusion, each such segment is necessarily contained in either $G_1$ or $G_2$, and therefore a corresponding segment can be found in $G'_1$ or $G'_2$ by the assumption on the fused graphs being equivalent. The corresponding segments can be put together to form a path in $G'_1 \oplus G'_2$ which connects the $i$-th pebble to the $j$-th without passing through pebbles from $I$. 
    \end{proof}

We will now use the congruence property described above to show that  languages definable in separator logic are closed under $k$-fusion. Suppose that $L_1$ and $L_2$ are languages of arity $k$ that are  definable in separator logic. 
     Choose the quantifier rank $r \in \set{1,2,\ldots}$ so that  both languages $L_1$ and $L_2$ are defined by formulas of separator logic which have  quantifier rank at most $r$. Define the  \emph{$r$-type} of a graph of arity $k$ to be its equivalence class with respect to $r$-equivalence.  By choice of  $r$,  for every $i \in \set{1,2}$, membership of a graph in $L_i$ depends only on its $r$-type.  
      
     A formula of separator which $k$ free variables and  quantifier rank $r$ will have $k+r$ variables, and therefore it can use only the separator predicates of arity at most $k+r$. In particular, up to logical equivalence there are finitely many equivalence classes, once the arity $k$ and the quantifier rank $r$ have been fixed.  This means that there are finitely many $r$-types of graphs of arity $k$.
      
     We can view graphs of arity $k$ as a commutative monoid, where the monoid operation is  $k$-fusion. By the congruence property from Claim~\ref{claim:fusion-congruence},  having the same $r$-type is a monoid congruence, and hence the $r$-types form a finite commutative  monoid. Also,  the monoid of $r$-types is aperiodic, in the sense that there is some threshold $m \in \set{0,1,2,\ldots}$ such that every monoid element $a$ satisfies
     \begin{align*}
     \myunderbrace{a \oplus \cdots \oplus a}{fusion of $m$\\ 
     \scriptsize copies of $a$} 
     \quad = \quad 
     \myunderbrace{a \oplus \cdots \oplus a}{fusion of $m+1$ \\ \scriptsize copies of $a$}.
     \end{align*}
In fact, aperiodicity holds not only for separator logic, but even for  the more expressive monadic second-order logic~\cite[Lemma 4.19]{bojanczykTwoMonadsGraphs2018}.

     For a graph $G$, define a  \emph{prime factor} of $G$ to be any graph $H$  that is obtained from $G$ by taking some non-port vertex $v$, and taking the subgraph of $G$ that is induced by the ports and those vertices which can be reached from $v$ by a path that does not visit ports.  Every graph is a fusion of its prime factors:
    \begin{align*}
    G = G_1 \oplus \cdots \oplus G_n.
    \end{align*}
    Lemma~\ref{claim:fusion-congruence} says that taking the $r$-type is a monoid homomorphism, from the monoid of graphs of arity $k$ equipped with fusion, to the monoid of $r$-types equipped with fusion.  It follows that for every graph of arity $k$, its $r$-type is uniquely determined by the following information: for every $r$-type $a$, what is the number of prime factors of $G$ that have $r$-type $a$, counted up to threshold the threshold $m$ from the definition of aperiodicity. We use the name \emph{$m$-profile} for this information. 
     The multiset of   prime factors of $G_1 \oplus G_2$ is the multiset union of the prime factors of $G_1$ and $G_2$, and therefore 
     \begin{align*}
     G \in L_1 \oplus L_2
     \end{align*}
      if and only if the multiset of prime factors in $G$ can partitioned into two multisets, the first multiset having  an $m$-profile that is consistent with $L_1$, and the second multiset having  an $m$-profile that is consistent with $L_2$. Whether or not this is possible can be deduced from the $2m$-profile of  the original graph. Finally, the $2m$-profile can be defined in separator logic.  This completes the proof that languages definable in separator logic are closed under fusion, and thus also the proof of Theorem~\ref{thm:star-free-is-logic}.

\section{Bounded pathwidth and aperiodicity}
\label{sec:pathwidth}
The second contribution of this paper is  an algebraic theory for star-free languages of graphs in the case of bounded pathwidth. Our goal is to give a version of \schutz's theorem about star-free languages and aperiodic monoids, see~\cite[Section 1]{Schutzenberger65}. \schutz's  theorem says that a word language is star-free if and only if it is recognized by a finite aperiodic monoid. We have already discussed aperiodic monoids in the previous section: a monoid is aperiodic if for every element $a$ there is some $m \in \set{1,2,\ldots}$ such that 
\begin{align*}
a^m = a^{m+1}.
\end{align*}
If the monoid is finite, as will be the case whenever we talk about aperiodic monoids, the number $m$ can be chosen independently of $a$.
Equivalently, a finite monoid is aperiodic if and only if it only contains trivial groups (a group contained in a monoid  is any subset of the monoid such that the monoid multiplication induces a group structure, with a group  identity that  is not necessarily  the same as the identity of  the monoid).

In this section, we prove a similar theorem for  graphs of bounded pathwidth, which characterizes the star-free graph languages (equivalently, graph languages  definable in separator logic) by a certain aperiodicity condition.   A corollary of our result is that there is an algorithm, which inputs a language $L$ of graphs of bounded pathwidth (e.g.~given by a formula of monadic second-order logic) and a pathwidth bound $k$, and decides whether or not there is a  star-free expression  (equivalently, a sentence of  separator logic) that coincides with $L$ on graphs of pathwidth at most $k$. The   usefulness of such an algorithm is perhaps debatable; but the  main point of the characterization is that it can be seen as further evidence that we have chosen the right  notion of star-free languages for graphs.

The  results of this section concern bounded pathwidth. 
Ideally, we would want to extend the results to more graphs, e.g.~for graphs of bounded treewidth. Such an extension seems to be  beyond the reach of our techniques. This is because  already for trees, i.e.~for graphs of treewidth 1,  finding an algebraic characterization of tree languages definable in first-order logic with descendant is a major open problem~\cite[Problem 3]{bojanczyk2015automata}.

     \paragraph*{Monoids.}
Our algebraic approach to graphs of bounded pathwidth is based on~\cite[Section 4.2.1]{bojanczykDefinabilityEqualsRecognizability2016a}: we view path decompositions as words over an alphabet that describes graph operations, and we recognize properties of path decompositions using homomorphisms  into finite monoids.  

We begin with a quick summary of monoids and their properties, some of which were already mentioned at the end of Section~\ref{sec:from-expressions-to-logic}. 
Recall that a \emph{monoid} is a set equipped with a binary associative operation, denoted by $a \cdot b$, together with a neutral element $1$ which satisfies $1 \cdot a = a = a \cdot 1$. A \emph{monoid homomorphism} is a function between two monoids that commutes with the monoid operation and preserves the neutral element. We say that a monoid homomorphism $\alpha : A \to B$ \emph{recognizes} a subset $L \subseteq A$ if elements of $A$ with equal value under $\alpha$ have equal membership status for $L$. We say that a monoid homomorphism $\alpha_1 : A \to B_1$ \emph{refines} a homomorphism $\alpha_2 : A \to B_2$ if equal values for $\alpha_1$ imply equal values for $\alpha_2$; this is the same as saying that $\alpha_2$ factors through $\alpha_1$ via some function of type $B_1 \to B_2$. (If $\alpha_1$ is surjective, then such a function, if it exists, is also a monoid homomorphism.) Every subset  $L$ of a monoid $A$ (we use the name \emph{language} for such subsets) is known to have a \emph{syntactic homomorphism}, i.e.~a homomorphism  which recognizes the language,  and which is refined by every other surjective  homomorphism that recognises the language.  One usually cares about languages contained in monoids that are finitely generated and free;  but in this paper will be interested in monoids that represent graphs of bounded pathwidth, and such monoids are not free.

\newcommand{\contexts}{\mathsf{Cont}}
\newcommand{\pw}{\mathsf{P}}
\newcommand{\mpw}{\mathsf M}

\subsection{Pathwidth and contexts}
Consider a graph (for the moment, without any ports). A \emph{path decomposition} of this graph   is defined to be a sequence of subsets of vertices, called \emph{bags},  such that: (1) every vertex appears in some bag and for   every  edge there is a bag that contains both endpoints of the edge; and  (2) for every vertex, the bags which contain this vertex form an interval in the sequence of bags  (i.e.~if a vertex appears  in two bags from the sequence, then it appears in all other bags between these two).  The width of a path decomposition is the maximal bag size, minus one. The \emph{pathwidth} of a graph is the minimal width of its path decompositions.

\begin{example}\label{ex:pathwidth-one}
    A graph has pathwidth 1 if and only if it can be obtained from a disjoint union of paths by adding dangling vertices, as in the following picture:
    \mypic{20}
    This property can be defined in separator logic. For higher pathwidth the situation is harder, and we believe that separator logic can  define pathwidth $k$ only for finitely many $k$.
\end{example}

For pathwidth, the appropriate notion of graphs with ports is  \emph{contexts}, which were called bi-interface graphs in~\cite[Definition 4.5]{bojanczykDefinabilityEqualsRecognizability2016a}. These  are like graphs with ports, except that there are two kinds of ports, called  left and right ports, as described in the following definition.

\begin{definition}[Context ] \label{def:contexts}  
    A \emph{context} of arity $k \in \set{0,1,\ldots}$  is defined to be a graph together with two 
     partial injective  maps
    \begin{align*}
     \text{left, right} : \set{1,\ldots,k} \to \text{vertices of the graph}.
    \end{align*}
    For $i \in \set{1,\ldots,k}$, the \emph{$i$-th left port} of the context is defined to be the image of $i$ under the left map; such a vertex is said to have \emph{left index $i$}.  Likewise we define the right ports. We require the following compatibility of left and right indices: if both the left and right indices of a vertex are defined, then these indices are equal.
\end{definition}


We write $u,v,w$ for contexts, because contexts play the role of words.
We allow the left or right ports in a context to be undefined, in contrast to the definition of graphs with ports which requires all ports to be defined. A graph without ports  can be viewed as a special case of a  context of any arity $k$, where all ports are undefined. 
Here is a picture of a context of arity 3, where the second right port is undefined, and the first left port is equal to the first right port:
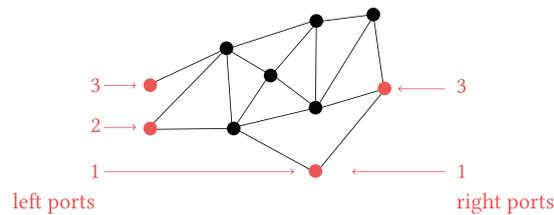
\mypic{5}
The compatibility of indices in the definition of contexts is meant to disallow contexts such as the following permutation gadget
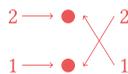
\mypic{13}

The idea behind contexts is that they  describe path decompositions, with the left ports being the contents of the leftmost bag, and the right ports being the contexts of the rightmost bag. 
Composition of contexts, which corresponds to concatenation of  path decompositions, is defined in the same way as fusion of graphs with ports, and is illustrated in the following picture:
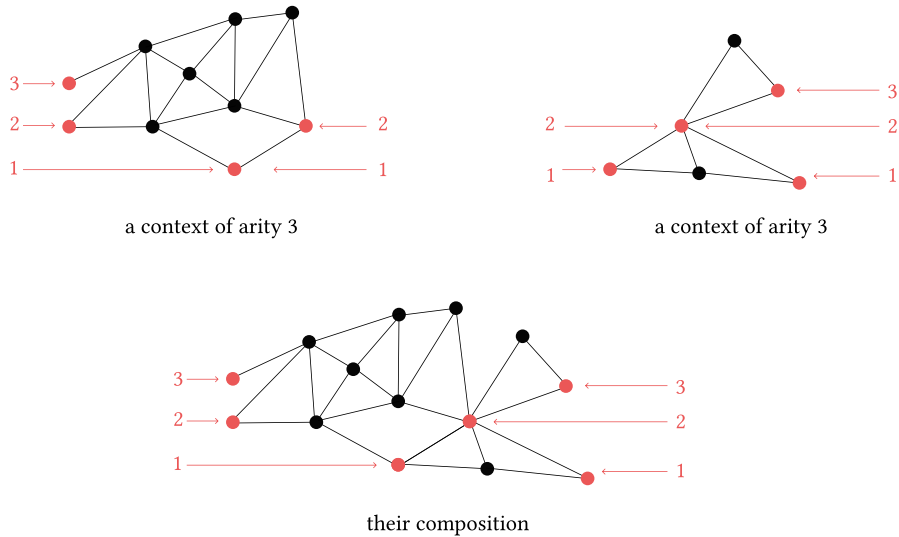
\mypic{6}
Composition of contexts is associative, and hence the context of fixed  arity $k$ form  a monoid.  Pathwidth is extended to contexts in the natural way:  a context has pathwidth at most $k$ if the underlying graph has a path decomposition of width at most $k$, where the leftmost bag contains all the left ports and the rightmost bag contains all the right ports. The following lemma is a straightforward reformulation of the definition of pathwidth; with the generating context playing the role of bags in a path decomposition. 
\begin{lemma}\label{lem:generators} \cite[Lemma 4.6]{bojanczykDefinabilityEqualsRecognizability2016a} A context has pathwidth at most $k$ if and only if it  can be generated, in the monoid of contexts of arity $k$, by contexts which have at most $k+1$ vertices. 
\end{lemma}
We will use the name \emph{$k$-generators} for the  generating contexts in the above lemma, i.e.~contexts of arity $k$ with at most $k+1$ vertices. We write $\pw_k$ for the monoid of contexts of arity $k$ that have pathwidth at most $k$; the above lemma says that this monoid is generated by the finite set of $k$-generators.
\begin{remark}
    In our algebraic approach to bounded pathwidth, we allow only one operation on contexts, namely composition. In the terminology of series parallel graphs; context composition corresponds to serial composition of graphs. A  parallel composition of contexts can also be defined,  call it \emph{fusion of contexts}, as illustrated in the following picture:
    \mypic{10}
Unfortunately, we cannot add fusion to our algebras, because contexts of given pathwidth are not closed under fusion. For example, in the picture above, the two fused contexts have pathwidth 1, but their fusion has pathwidth 2.
In fact, adding fusion to the monoid  of pathwidth $k$  would allow us to generate all graphs of treewidth $k$, since we can view a graph with $k$ ports  as a context where only the left $k$ ports are defined, and the right ports are all undefined. Using fusion on such contexts would allow us to simulate arbitrary tree decompositions of width $k$. 
\end{remark}

\subsubsection{Recognizable languages of pathwidth $k$}
We will be interested in classifying recognizable languages of given pathwidth, i.e.~languages which are recognized by a monoid homomorphism from the monoid $\pw_k$ to a finite monoid.  We will want to know which of these languages can be defined by star-free expressions, or equivalently, using separator logic, in the sense that is described by the following lemma.

\begin{lemma}\label{lem:star-free-pathwidth}
    Let $L$ be a set of graphs without ports, which contains only graphs of pathwidth at most  $k \in \set{1,2,\ldots}$. The following conditions are equivalent:
    \begin{enumerate}
        \item there is a formula of separator logic that agrees with $L$ on graphs of pathwidth at most $k$;
        \item there is a star-free expression that agrees with $L$ on graphs of pathwidth at most $k$;
        \item \label{it:star-free-in-monoid} if we view $L$ as a set of contexts where all ports are undefined, then $L$ belongs to the least class of subsets of the monoid $\pw_k$ that contains all finite sets, and  is closed under the monoid operation (context composition), as well as Boolean combinations (with complementation relative to the monoid $\pw_k$).
    \end{enumerate}
\end{lemma}
\begin{proof}
    The equivalence of the first two items follows immediately from Theorem~\ref{thm:star-free-is-logic}. The equivalence of the third item with the first two is proved using a straightforward induction.
\end{proof}

We will say that a language of graphs without ports is \emph{star-free over pathwidth at most $k$}, if it satisfies any of the equivalent conditions from the above lemma. In the proofs from this section, it will also be convenient to talk about star-free languages of contexts where the ports are allowed to defined; for such languages we use the third item of the lemma. Also, for contexts that have possibly defined ports one could also use separator logic; this would be done by associating to each context a model that has a (possibly undefined) constant for each left and right port.

\begin{remark}
    In the first two items of Lemma~\ref{lem:star-free-pathwidth}, we do not require the formula or star-free expression to check if the graph has pathwidth at most $k$.  An alternative definition, which would require the formula or expression to only be true in graphs of pathwidth at most $k$. This alternative definition   seems to be overly restrictive. The reasons is that, as mentioned in Remark~\ref{ex:pathwidth-one}, we believe that there is some $k$ that satisfies: (*) there is no formula of separator logic that defines pathwidth at most $k$ as a subset of the  graphs of pathwidth at most $k+1$.  In fact, this belief could even be confirmed by running an algorithm: as will follow from Corollary~\ref{cor:decide-separator} later on in this section, there is an algorithm that inputs $k$ and  answers whether or not it satisfies (*).
\end{remark}

The purpose of this section is to give an algebraic characterization of those recognizable languages of pathwidth at most $k$ which are star-free. Like \schutz's characterization for words, our characterization will lead to an algorithm that checks if a language can be defined. Another similarity with \schutz's characterization is that we use aperiodic monoids. However, there is a certain twist, which is described in the following example.

\begin{example}[Periods from reachability]\label{ex:reach-aperio}
    Let $L \subseteq \pw_2$ be the set of contexts of arity 2 that   admit a  path from left port 1 to right port 1. This language is easily seen to be  star-free. However, the syntactic monoid of this language contains a copy of the two-element group. Indeed,  
    let $w$ be the following context:
    \mypic{11}
    If we consider powers $w^n$ of this context, then the syntactic homomorphism of $L$  will assign different values to even and odd powers, because left port 1 can be connected to right port 1 in $w^n$ if and only if $n$ is even.
\end{example}

As the above example shows, our algebraic characterization of star-free languages  cannot simply say that the recognizing monoid is aperiodic. However, it turns out that the reachability periods described in the above example are  the only kind of periods  that are permitted for star-free languages. To formalize what we mean by ``reachability period'', we use the following homomorphism, which was introduced under the name of \emph{abstraction} in~\cite[Section 4.2.1]{bojanczykDefinabilityEqualsRecognizability2016a}.

\begin{definition}[Reachability homomorphism]\label{def:reachability-homomorphism}
    An \emph{inner path} in a context is defined to be  a path that does not use ports, with the possible exception of its source and target vertices. 
      Define the \emph{reachability homomorphism}, for contexts of arity $k \in \set{1,2,\ldots}$,  to be the function which maps a context of arity $k$ to the following information: (a) which indices $i \in \set{1,\ldots,k}$ describe persistent ports; and (b) for which pairs 
      \begin{align*}
      (x,y) \in \set{\text{left, right}} \times \set{1,\ldots,k}
      \end{align*}
      is there an inner path from port $x$ to port $y$. 
\end{definition}

    The reachability homomorphism is indeed a homomorphism, i.e.~if we know the values of this homomorphism for two contexts $w$ and $v$, then we also know the value of this homomorphism for the composition $w \cdot v$. An example of a language that is recognized by the reachability homomorphism is the contexts which admit an inner path from  left port 1 to right port 1.
  
We are now ready to present the main result of Section~\ref{sec:pathwidth}. Recall that an \emph{idempotent} in a monoid is an element $e$ that satisfies $e = e \cdot e$.

\begin{theorem}\label{thm:aperiodicity} A language $L \subseteq \pw_k$ is star-free if and only if it  is recognized by a  homomorphism $\alpha : \pw_k \to A$ into a finite monoid $A$  such that every context $w \in \pw_k$ satisfies the following implication:
        \begin{align*}
            \myunderbrace{\beta(w)}{here $\beta$ is the \\
            \scriptsize reachability homomorphism } \text{ is idempotent} \qquad \Rightarrow \qquad \myunderbrace{\text{$\alpha(w)$ is aperiodic}}{an element $a \in A$ is called \emph{aperiodic} if \\  \scriptsize   $a^m = a^{m+1}$ holds for some $m \in \set{1,2,\ldots}$}.
            \end{align*}    
\end{theorem}

A homomorphism $\alpha$ that satisfies the implication in the above theorem is called  \emph{aperiodic modulo reachability}. 
\begin{example}[Two disjoint paths]\label{ex:disjoint-paths}
    Let $L \subseteq \pw_2$ be the contexts with the following property: there exist two vertex disjoint paths, one path from the first left port to first right port, and another path from the second left port to the second right port. We will show that this language is not star-free, because its syntactic homomorphism is not aperiodic modulo reachability. (If any recognizing homomorphism is aperiodic modulo reachability, then the syntactic homomorphism is as well.) Indeed, let $w$ be the following context 
    \mypic{14}
    This context is mapped to an idempotent by the reachability homomorphism: all powers of this context have the property that all ports are distinct, and every two ports can be connected by an inner path. However, we will show that the image of this  context under the syntactic homomorphism of the language $L$ is not aperiodic, and therefore the language is not star-free.
    Indeed, if $m$ is even, then the context $w^m$ belongs to the language,  as explained with two paths (blue and orange) in the following picture for $m=4$:
    \mypic{16}
    On the other hand, if $m$ is odd, then it is impossible to find two such paths.  This example shows that separator logic cannot express the existence of two vertex disjoint paths with given sources and targets; this result was shown independently in~\cite[Corollary 3.10]{schrader2021firstorder}. Based on this impossibility, Schrader et al.~introduce an extension of separator logic that can talk about disjoint paths~\cite[Section 4]{schrader2021firstorder}. One could imagine that our characterization from Theorem~\ref{thm:aperiodicity} could be modified to handle this extension, possibly by considering other  homomorphisms than the reachability homomorphism used for $\beta$.
\end{example}

The characterization in Theorem~\ref{thm:aperiodicity} is  effective, as expressed in Corollary~\ref{cor:decide-separator} below. In the corollary, we use monadic second-order logic to define properties of contexts. This Courcelle's \mso\!$_1$ logic: it can quantify over vertices, sets of vertices, and it is equipped with a binary edge relation.
\begin{corollary}\label{cor:decide-separator}
    Given  $k \in \set{1,2,\ldots}$ and  a formula $\varphi$ of monadic second-order logic, one can decide if there is a formula of separator logic that is equivalent to $\varphi$ on graphs of pathwidth at most $k$.
\end{corollary}
\begin{proof}
    We assume that the formula describes a property of graphs without ports, but the same proof would also work with ports.
    By Courcelle's Theorem, one can compute a monoid homomorphism from $\pw_k$ to some finite monoid, which recognizes the language of $\varphi$ on graphs of pathwidth at most $k$. Next, using a natural fixpoint procedure (called Moore's minimization algorithm), one can minimize the monoid, and thus we can assume that $h$ is the syntactic homomorphism. As mentioned in Example~\ref{ex:disjoint-paths}, if any recognizing homomorphism is aperiodic modulo reachability, then the syntactic homomorphisms has this property as well. Therefore, by Theorems~\ref{thm:star-free-is-logic} and \ref{thm:aperiodicity}, the formula $\varphi$ is equivalent to a formula of separator logic on graphs of pathwidth at most $k$ if and only if the syntactic homomorphism $\alpha$ is aperiodic modulo reachability. The latter property can be checked effectively, by an exhaustive search of the  product of the syntactic monoid  with the   target monoid of the reachability homomorphism.
\end{proof}

In the above corollary, the language is represented by a formula of monadic second-order logic (one could also use the extension of this logic which allows modulo counting). From the point of view of computational complexity, a bottleneck in the algorithm is Courcelle's Theorem, where the conversion from logic to a monoid requires a tower of exponentials. This tower can be avoided if  the language is represented by  a homomorphism into a finite monoid; in this case the algorithm runs in polynomial time.

The rest of Section~\ref{sec:pathwidth} is devoted to proving Theorem~\ref{thm:aperiodicity}. 

The easier implication in the theorem is that  if a language is star-free, then it is recognized by a  homomorphism that  is aperiodic modulo reachability.  For this proof we use separator logic instead of star-free expressions, and  a standard a \ef argument, which is only sketched. See Remark~\ref{remark:not-schutz} for why we do not follow \schutz's original proof by using star-free expressions.

As mentioned before, we can use separator logic to define properties of contexts; this is done by a viewing a context as a model in the same way as a graph with ports (the differences being that there are left and right ports, and some can be undefined). 
For a quantifier rank $r \in \set{0,1,\ldots}$ define the \emph{$r$-type} of a context $w$ of arity $k$ 
to be the set of sentences of separator logic that have quantifier rank at most $r$ and  are true in the context. A standard compositionality argument, similar to Claim~\ref{claim:fusion-congruence}, shows that for every $r$, the function  which maps a context to its $r$-type is a homomorphism into a finite monoid. Furthermore, if $r$ is the quantifier rank of the formula that defines $L$,  then this homomorphism recognizes the language $L$.  The following lemma shows that this  homomorphism is aperiodic modulo reachability, thus proving the easier implication in Theorem~\ref{thm:aperiodicity}.

\begin{lemma}
    Suppose that $w$ is a context which is mapped to an idempotent by the reachability homomorphism. For every quantifier rank $r$, there is some  $m \in \set{1,2,\ldots}$ such that the contexts $w^m$ and $w^{m+1}$  have the same $r$-type. 
\end{lemma}
\begin{proof}
    Take $m$ to be $3^r$.
     A vertex of $w^m$
     has a representation  as a pair 
    (vertex of $w$, number in $\set{1,\ldots,m}$), with the number indicating which copy of $w$ inside $w^m$ is used. This representation is not  unique, due to ports being   shared across consecutive copies. To ensure uniqueness,  define the \emph{canonical representation} to be the representation which uses the smallest number on the second coordinate. This canonical representation is unique, and its two coordinates are called the \emph{offset} (a vertex of $w$) and \emph{index} (a number in $\set{1,\ldots,m}$), respectively. Similarly we define the index and offset for vertices of $w^{m+1}$.  
    Using a standard argument, see~\cite[proof of Theorem IV.2.1]{straubingFiniteAutomataFormal1994}, one can show that Duplicator has a strategy in the $r$-round \ef game on $w^m$ and $w^{m+1}$ which ensures  that once all $r$ rounds have been played and $i, j \in \set{1,\ldots,r}$ are pebble names, then: (1) the offsets of pebble $i$ in $w^m$ and $w^{m+1}$ are the same;
        and (2) the order on indices of pebbles $i$ and $j$ are the same in $w^m$ and $w^{m+1}$; and (3)
         the indices of pebbles $i$ and $j$ differ by a number that is not 1,    both in $w^m$ and   $w^{m+1}$. 
    Because of the above assumptions, and since one or more copies of $w$ are equivalent with respect to the reachability homomorphism, it follows that the two pebblings satisfy the same reachability predicates. This, in turn, means that Duplicator's strategy was winning.
\end{proof}

This completes the proof of the easier implication in Theorem~\ref{thm:aperiodicity}. We finish this section with a discussion on why we proved the easier implication using separator logic instead of star-free expressions.

\begin{remark}[\schutz product fails for graphs]
    \label{remark:not-schutz}
    When proving that  every star-free language is recognized by an aperiodic monoid~\cite[Remark 2]{Schutzenberger65}, \schutz does not use logic, but a monoid construction that is now called the \emph{\schutz product}. We describe this construction here and explain why it does not work for graphs of bounded pathwidth. Suppose that we have two monoid homomorphisms $\set{h_i : \Sigma^* \to M_i}_{i \in \set{1,2}}$ with the same domain. If we want to recognize the concatenation of two languages, recognized by $h_1$ and $h_2$ respectively, it is natural to consider the function  
\begin{align*}
  H : \Sigma^* \to M =  \myunderbrace{M_1 \times M_2 \times \powerset(M_1 \times M_2)}{\schutz product of $M_1$ and $M_2$}
\end{align*}
which maps a word $w$ to: its values under $h_1$ and $h_2$, as well as the set of all  pairs $(h_1(w_1),h_2(w_2))$ that range over decompositions $w = w_1 \cdot w_2$. It is not hard to see that $H$ is a monoid homomorphism, using a suitably defined monoid structure on $M$, this homomorphism recognizes every language obtained by concatenating a language recognized by $h_1$ with a language recognized by $h_2$, and finally,  if both $M_1$ and $M_2$ are aperiodic, then so is $M$. These properties imply that word languages recognized by aperiodic monoids are closed under concatenation. Unfortunately, the same strategy does not work for graphs: if we replace $\Sigma^*$ by the  monoid $\pw_k$, then $H$ is no longer a monoid homomorphism. The difference is that words have the following property that is not shared by contexts: if we want to split a word $w_1 \cdot w_2$ into to parts, then this is done by either splitting $w_1$ into two parts and keeping $w_2$ intact, or by keeping $w_1$ intact and  splitting $w_2$ into two parts. It is possible that there is a way to fix the \schutz product for graphs of bounded pathwidth, but we do not pursue this here.
\end{remark}

\subsection{From aperiodicity to separator logic}
The rest of Section~\ref{sec:pathwidth} is devoted to the harder implication of Theorem~\ref{thm:aperiodicity},  
which says that  if a language of contexts  is recognized by some homomorphism  that is aperiodic modulo reachability, then the language is star-free.
When proving the harder implication, we use the same strategy as  \schutz's original proof for words. There are, however,   extra arguments that are specific to graphs. These arguments revolve mainly around the notion of bridges, which is described below.

\subsubsection{Bridges}
\label{sec:bridges}
We partition the edges of a context into  \emph{inner components} as follows:  two edges of  a context are in the same inner component if they can be connected by using non-port vertices and other edges.   In other words, two edges are in the same inner component if they both appear in one inner path. Here is a picture of a context with its  inner components:
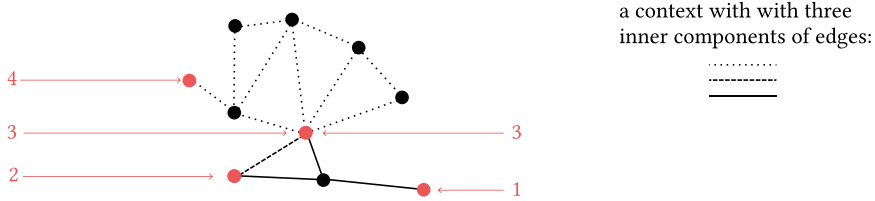
\mypic{9}

We say that an inner component is \emph{incident} to some port if some edge from the inner component is incident to that port. 
\begin{definition}
    [Bridge] \label{def:bridges} A vertex in a context is called \emph{persistent} if it is both a left and right port.  An inner component in a context is called  a \emph{bridge} if, after removing the persistent vertices, the component is  incident to a left and a right port.
\end{definition}

For example, in the most recent picture there is  exactly one bridge, namely the inner component drawn using solid lines \inlinepic{22}. This inner component is incident to the second left port and the first right port, both of which are not persistent.  The remaining inner components are not bridges, because they are not incident to the first right port, which is the only right port that is not persistent.

\begin{example}\label{ex:periodic-with-two-bridges}
    Consider  the problematic context 
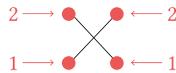
\mypic{11}
from Example~\ref{ex:reach-aperio}, which is periodic with respect to the reachability homomorphism. In this context, there are no persistent vertices, and two bridges, namely each of the two edges is a singleton bridge.  Hence, this problematic context has two bridges.
\end{example}

In the Example~\ref{ex:periodic-with-two-bridges}, two bridges were used to get a period  under the reachability homomorphism. It turns out that this is the case in general.

\begin{lemma}
    \label{lem:technical-blob} If a context  has at most one bridge,  then its image  under the reachability homomorphism is aperiodic. 
\end{lemma}
\begin{proof}
    Let $w$ be a context of arity $k$.
    The reachability homomorphism tells us which ports are defined, which ports are persistent, and  which pairs of source and target ports
    can be connected by an inner path. The first two pieces of information, about which ports are defined and persistent, does not change as we take powers $w^n$ of a context. The interesting part is the connectivity using inner paths. The intuitive idea is that this connectivity is aperiodic, since all paths that traverse the context must pass through the same bridge, and therefore there is no place for periodic swapping. A more precise analysis is given below.

    Suppose that we want to know if two ports  are connected by an inner path in a power $w^n$.
    Suppose first that the two ports are on the same side, say both are left ports. It is easy to see that if two left ports can be joined by an inner path  in $w^n$, then they will also be joined by the same inner path in $w^{n+1}$.   Therefore, the  answer to the question in the claim can only go from ``no'' to ``yes'', and so it eventually stabilizes. A symmetric argument works when both the source and target are right ports. 
    We are left with the case when the source  is a left port and the target is a right port. 
    Also, none of these ports are persistent, since otherwise they would both be left ports, or they would both be right ports, and we could use the arguments above. This remaining  case is dealt with by the following claim, which completes the proof of the lemma. 
    \begin{claim}
        Let  $s , t \in \set{1,\ldots,k}$ be such that  neither the $s$-th left port nor the $t$-th right port are persistent. All but finitely many $n \in \set{1,2,\ldots}$ lead to the same answer to the question:
                    does  $w^n$ have an inner path from the $s$-th left port to the $t$-th right port?
    \end{claim}
    \begin{proof}
        For  numbers $n \in \set{1,2,\ldots}$ and $i \in \set{1,\ldots,n}$, one can speak of the $i$-th copy of $w$ inside $w^n$ in the following sense:
        \mypic{19} These copies are not disjoint, since port vertices appear in several copies.
         The above picture does not use persistent ports, and indeed this can be assumed without loss of generality (which will simply the reasoning). Indeed,  the answer to the question in the claim does not change if replace $w^n$ by either of the following two contexts, which are furthermore equal to each other: the context obtained from $w^n$ by  removing persistent ports; or the $n$-th power of the context obtained from $w$ by removing persistent ports. From now on, we assume that $w$ has no persistent ports.
        
        Define $S^n_i$ to be the vertices $x$ of $w$ such that $w^n$ admits an inner path from its $s$-th left port to the $i$-th copy of $x$. We will use the assumption on having a unique bridge to show that 
        \begin{align}\label{eq:bridge-decreases}
        S^{n}_{i} \subseteq S^n_{i-1}.
        \end{align}
        To prove the above inclusion, consider an inner path $\pi$ in $w^n$ that begins in the $s$-th left port and reaches the $i$-th copy of $x$ for $i > 2$.  The key observation is that, since $w$ has a unique bridge, then for every $j \in \set{1,2,\ldots, i-1}$, the path $\pi$ must use at least one edge from the $j$-th copy of the unique bridge; otherwise it would be impossible to reach the $i$-th copy of $w$.  Among all edges used by  $\pi$, consider the last one that is   from the second copy of the unique bridge in $w$, and let $\sigma$ be the suffix of $\pi$ that begins in this edge.  This  suffix  $\sigma$ cannot use any vertices from the first copy of $w$, and therefore we can shift $\sigma$ by one copy to the left, and get a legitimate inner path in $w^{n}$, call it $\sigma_{\leftarrow}$. The first edge used by  $\sigma_{\leftarrow}$ is from the first copy of the unique bridge. The path $\pi$ must use some edge from the first copy of the unique bridge, and this edge can be connected by an inner path to the first edge used by $\sigma_{\leftarrow}$, since they both belong to the first copy of the unique bridge. Summing up, we have shown that the target vertex of $\pi$ can be shifted by one copy to the left, without affecting reachability via an inner path from the $s$-th left port, thus proving~\eqref{eq:bridge-decreases}.

        By the same kind of monotonicity argument as we used when talking about paths connecting two left ports,  one can show that there is some threshold  such that 
        \begin{align}
            S_i^{n} = S_i^{n+1} 
        \end{align}
        holds for all  $i < n$ such that the distance from $n$ to $i$ exceeds the threshold.
        For $i \in \set{1,2,\ldots}$,  let us write $S_i$ for the value of $S^n_i$ for some (equivalently, every) $n$ whose distance to $i$ exceeds the threshold. By~\eqref{eq:bridge-decreases}, we know that the sequence $S_1, S_2,\ldots$ is decreasing, and therefore it must eventually reach some stable value, which we denote by $S$. 
         Define $T^n_i$ in the same way as $S^n_i$, except that the $t$-th right port  is used instead of the $s$-th left port. By the same argument as for $S^n_i$, there is some stable value $T$ such that  $T= T^n_i$
        holds for all $i<n$ such that the distance from $1$ to $i$  exceeds the threshold. Consider now a number $n$ that is bigger than the sum of both thresholds (the threshold for $S$ and the threshold for $T$). If we take $i$ to be half of $n$, rounded up, then 
        \begin{align*}
        S = S^n_i \qquad \text{and} \qquad T = T^n_i.
        \end{align*}
        If the sets $S$ and $T$ have nonempty intersection, then $w^n$ admits an inner path from the $s$-th left port to the $t$-th right port; if the sets are disjoint then there is no such path. In other words, the fixed sets $S$ and $T$  determine the answer to the question in the claim for all large enough $n$, thus proving the claim.
    \end{proof}
\end{proof}
The purpose of the above lemma is to get the following corollary.
\begin{corollary}\label{cor:bridges-aperiodic}
    If  $w$ is a context with   at most one bridge,  and $\alpha : \pw_k \to A$ is a homomorphism   that  is aperiodic  modulo reachability, then $\alpha(w)$ is aperiodic.
\end{corollary}
\begin{proof}
    This is an almost immediate application of Lemma~\ref{lem:technical-blob} and the definition of aperiodicity modulo reachability, with one caveat. To use the lemma, we would like to have the implication 
    \begin{align}\label{eq:weaker-assumption}
        \text{$\beta(w)$ is aperiodic}\qquad \Rightarrow \qquad \text{$\alpha(w)$ is aperiodic}, 
        \end{align}
        since aperiodicity of $\beta(w)$ is what we get from Lemma~\ref{lem:technical-blob}. On the other hand, the  definition of aperiodicity modulo reachability has an implication with a stronger assumption, namely: 
    \begin{align}
        \label{eq:stronger-assumption}
        \text{$\beta(w)$ is idempotent}\qquad \Rightarrow \qquad \text{$\alpha(w)$ is aperiodic}.
        \end{align}
    However, we will show that the two implications~\eqref{eq:weaker-assumption} and~\eqref{eq:stronger-assumption} are the same. Indeed, suppose that the assumption of~\eqref{eq:weaker-assumption} holds, namely $\beta(w)$ is aperiodic. This means that for sufficiently high powers $n$, $\beta(w^n)$ is an idempotent, and therefore $\alpha(w^n)$ is aperiodic. In other words, we have established that all sufficiently high powers of $\alpha(w)$ are aperiodic. To conclude that $\alpha(w)$ is aperiodic, and thus prove~\eqref{eq:weaker-assumption}, we use the following observation about finite monoids:
    \begin{itemize}
        \item[(*)] If $a$ is an element of a  monoid such that some two consecutive powers $a^n$ and $a^{n+1}$ are aperiodic, then also $a$ is aperiodic.
    \end{itemize}
    To see why (*) is true, we observe that all  sufficiently large  $m$ satisfy
    \begin{align*}
    a^{m+n+1} \qquad  \myunderbrace = {because $a^{n+1}$ \\
    \scriptsize is aperiodic} \qquad
     a^m 
     \qquad  \myunderbrace = {because $a^{n}$ \\
    \scriptsize is aperiodic} \quad
       a^{m+n},
    \end{align*}
    which implies that all sufficiently large $m$ satisfy $a^m= a^{m+1}$, thus establishing aperiodicity of $a$. 
\end{proof}

\subsubsection{The induction}
Having established aperiodicity of contexts with at most one bridge, we  return to the proving the harder implication of Theorem~\ref{thm:aperiodicity}, which says that if a language is  recognized by a homomorphism that is aperiodic modulo reachability, then it is star-free. 
Consider a  homomorphism
\begin{align*}
\alpha : \pw_k \to A,
\end{align*}
which is aperiodic modulo reachability. We want to show that every language recognized by $\alpha$ is star-free. Here,  star-free is understood in the sense of  item 3 from Lemma~\ref{lem:star-free-pathwidth}, i.e.~the language can be generated in the monoid $\pw_k$, from finite languages, by using Boolean operations and concatenation. Without loss of generality, we can assume that $\alpha$ refines the reachability homomorphism. Indeed, if  we  add the outputs of the reachability homomorphism to a homomorphism $\alpha$, then the new homomorphism will refine the reachability homomorphism, will  still be aperiodic modulo reachability, and will recognize all languages recognized by the original homomorphism $\alpha$.

Fix a homomorphism $\alpha$ for the rest of this section; this homomorphism is assumed to be  aperiodic modulo reachability, and to refine the  reachability homomorphism. Also without loss of generality, we can assume that $\alpha$ is  surjective. We will show that every language recognized by $\alpha$ is star-free.  
We say that a monoid element  $a \in A$ is \emph{star-free} if the inverse image $\alpha^{-1}(a)$ is star-free. We will show that every element of the monoid $A$ is star-free; this will immediately imply that every language recognized by $\alpha$ is star-free, since every such language  is a finite union of inverse images $\alpha^{-1}(a)$. 

\paragraph*{Green's relations.}
In the proof, we will use Green's relations, so we begin with a quick summary of these relations and their basic properties. Green's relations are  three pre-orders on elements of the monoid $A$, called \emph{prefix}, \emph{suffix} and \emph{infix}, which are defined as follows in terms of monoid ideals (the  terminology of infixes, prefixes and infixes taken from~\cite[Section 1.2]{bojanczyk_recobook}):
\begin{align*}
\myunderbrace{a \in b \cdot A}{$b$ is a prefix of $a$}
\qquad
\myunderbrace{a \in A \cdot b}{$b$ is a suffix of $a$}
\qquad
\myunderbrace{a \in A \cdot b \cdot A}{$b$ is an infix of $a$}.
\end{align*}
All of these relations are pre-orders, i.e.~they are transitive and reflexive, but not necessarily anti-symmetric. Two monoid elements are called \emph{prefix equivalent} if they are prefixes of each other (in other words, they generate the same right ideals $a \cdot A = b \cdot A$), likewise we define suffix equivalence and infix equivalence. Define the $\Hh$-class of a monoid element to be the intersection of its prefix class with its suffix class. We will be particularly interested in $\Hh$-classes which contain an idempotent. The following lemma, see~\cite[Lemmas~1.11--1.14]{bojanczyk_recobook}, sums up the properties of Green's relations that will be used in our proof. 

\begin{lemma}[Green's Relations Lemma]\label{lem:green}
    The following properties hold in every finite monoid:
    \begin{enumerate}
        \item \label{it:eggbox} If $a$ is a prefix of $b$, and $a$ is infix equivalent to $b$, then $b$ is a prefix of $a$; likewise for suffixes.
        \item \label{it:h-classes-same-size}  All $\Hh$-classes contained in the same infix class have the same size.
        \item \label{it:h-class-with-idempotent}  If an $\Hh$-class contains an idempotent, then it is a group.
    \end{enumerate}
\end{lemma} 

The groups in item~\ref{it:h-class-with-idempotent} can be trivial, i.e.~have only one element.
\paragraph*{The induction.} Having described Green's relations, we resume the proof that every element of $A$ is star-free. 
The proof  is by induction on the position of the element  in the infix ordering. Consider some infix class $J \subseteq A$, and assume that we have already shown star-freeness for every monoid element that is a strict infix of some (equivalently, every) element of $J$. We will show that also all monoid elements in $J$ are star-free. The proof will consider  two cases, which are identified in the following lemma.

\begin{lemma} \label{lem:must-be-h-trivial} Let $\alpha : \pw_k \to A$ be a surjective monoid homomorphism, which refines the reachability homomorphism,  and which is aperiodic modulo reachability.
    Every infix class $J \subseteq A$ satisfies at least one of the following conditions:
    \begin{enumerate}
        \item every element of $J$ is aperiodic; or 
        \item \label{it:at-least-two-bridges} every context in $\alpha^{-1}(J)$ has at least two bridges.
    \end{enumerate}
\end{lemma}
\begin{proof} To prove the lemma, we will show that if  some context in $\alpha^{-1}(J)$ has at most one bridge, then  every element of $J$ is aperiodic. We begin with the following claim, which shows that the number of persistent ports is an invariant of an infix class.

    \begin{claim}\label{claim:number-of-persistent-ports-invariant}
        All contexts in $\alpha^{-1}(J)$ have the same number of persistent ports.
    \end{claim}
    \begin{proof}
         Since $\alpha$ refines the reachability homomorphism, and the reachability homomorphism stores information about which ports are persistent, it follows that the image under $\alpha$ determines the number of persistent ports.  
        If we compose two contexts, then the context can only have fewer persistent ports. It therefore follows that the invariant ``number of persistent ports'' can only decrease when composing contexts; and therefore this invariant must be constant within a single infix class.
    \end{proof}
    
    The following claim shows that having at most one bridge is also an invariant of the infix class.

    \begin{claim}
        Let $v$ and $w$ be contexts with the same number of persistent ports, such that $v$ is an infix of $w$. If  $v$ has at most one bridge, then so does $w$.
    \end{claim}
    \begin{proof}
        Let $f$ be the function which maps vertices of $v$ to their corresponding vertices in $w$, when $v$ is viewed as an infix of $w$. This function maps inner paths of $v$ to inner paths of $w$. 
        Since $v$ and $w$ have the same number of persistent ports,  the function $f$ does not change persistence:
        \begin{align*}
            \text{$x$ is a persistent vertex in $v$} \qquad \Leftrightarrow \qquad \text{$f(x)$ is a persistent vertex in $w$}.
            \end{align*}
            Using these two observations, we prove the claim. 
        Consider an inner path $\pi$ in $w$ which connects some non-persistent left port with some non-persistent right port. This path must necessarily contain a segment  which is the image, under $f$, of some path $\sigma$ in $v$ that goes from a left port of $v$ to a right port of $w$. Since $f$ does not change persistence, it follows that $\sigma$ must use a vertex from the unique bridge in $v$. Therefore, $\pi$ must use a vertex from the image of this unique bridge under $f$, and this image is contained in a unique component of $w$.  
    \end{proof}

    Since we have assumed that some context in $\alpha^{-1}(J)$ has  at most one bridge, we can use the above claim to  conclude that every context in $\alpha^{-1}(J)$ has at most one bridge. By Corollary~\ref{cor:bridges-aperiodic} we know that all elements of  $J$ are  aperiodic. 
\end{proof}

Using the above lemma, it is enough to show the induction step for an infix class which satisfies one of two cases in the lemma.

\subsubsection{Only aperiodic elements}
\label{sec:no-idempotents}
We begin with the case where every element of $J$ is aperiodic. Here, we use the same proof as \schutz's proof for words. We further split this case into two sub-cases, depending on whether $J$ contains an idempotent. 

\paragraph*{Has no idempotents.} 
Assume first that $J \subseteq A$ is an infix class without any idempotent. (Actually, once we assume that there are no idempotents, then the assumption on having only aperiodic elements will not be used.) In monoid theory,  an infix class without idempotents is called \emph{non-regular}, and such infix classes are known to be simple because they necessarily decompose into strictly smaller infix classes, as explained in the following lemma.

\begin{lemma} If $J$ is an infix class without idempotents, then every context  in $\alpha^{-1}(J)$ can be decomposed as $w_1 w_2 w_3$ such that $w_2$ is a $k$-generator and both $w_1$ and $w_2$ have images under $\alpha$ that are  strict infixes of $J$.
\end{lemma}
\begin{proof}
    Consider a context $w$ with image in $J$, and view it as a word whose letters are $k$-generators.  Let $w_1$ be the longest prefix of this word whose  image is not in $J$ (and therefore this image is a strict infix of $J$), let $w_2$ be the letter just after $w_1$, and let $w_3$ be the rest of the word. If $w_3$ would have image in $J$, then $J \cdot J$ would contain an element of $J$, which cannot happen in an infix class without an idempotent,  see~\cite[Corollary 2.25]{PinMPRI}.
    \end{proof}

The above lemma, together with the induction assumption, implies the sufficient condition in the following lemma, with $\ell =3$,  and thus proves star-freeness of every element in $J$. 

\begin{lemma} \label{lem:reuse-bounded-infixes} The following condition is sufficient for star-freeness of $a \in A$: every monoid element that is a strict infix of $a$ is star-free, and  there is some $\ell \in \set{1,2,\ldots}$ such that every  context in $\alpha^{-1}(a)$ can be decomposed as $w_1 \cdots w_\ell$ such that each $w_i$ is either a  $k$-generator, or $\alpha(w_i)$  is a strict infix of~$a$. 
\end{lemma}
\begin{proof}
The contexts in $\alpha^{-1}(a)$ are described by the star-free expression
    \begin{align*}
    \bigcup_{\substack{a_1,\ldots,a_\ell\\ a = a_1 \cdots a_\ell}} L_{a_1} \cdots L_{a_\ell}
    \end{align*}
    where $L_b$ is defined to be the following star-free language:
    \begin{align*}
    L_b = \begin{cases}
        \alpha^{-1}(b) & \text{when $b$ is a strict infix of $b$};\\
        \text{$k$-generators in }\alpha^{-1}(b) & \text{otherwise}.
    \end{cases}
    \end{align*}
\end{proof}

\paragraph*{Has idempotents.}
Consider now an infix class $J$ where all elements are aperiodic, and which contains at least one idempotent. 
The infix class contains at least  one $\Hh$-class which is a group, since it contains an idempotent (see item~\ref{it:h-class-with-idempotent} of Lemma~\ref{lem:green}). This $\Hh$-class must have size one, since a group with only aperiodic elements is necessarily of size one. Finally, since  all $\Hh$-classes in a given infix class have the same size, by item~\ref{it:h-classes-same-size} of Lemma~\ref{lem:green}, it follows that all $\Hh$-classes are trivial.  

We begin with the following lemma, which says that  if we are given a context whose value under $\alpha$ is promised to be in $J$, then we can use star-free expressions to determine that value.  
 \begin{lemma}\label{lem:approx-j} Every $a \in J$ is star-free over $\alpha^{-1}(J)$ in the following sense: there is a star-free expression that coincides with $\alpha^{-1}(a)$ over contexts from $\alpha^{-1}(J)$.
 \end{lemma}
 \begin{proof}
    Let  $w \in \alpha^{-1}(J)$. Let $P \subseteq J$ be the prefix class of $a$. By  item~\ref{it:eggbox} of Lemma~\ref{lem:green}, we know that for elements of $J$, belonging to $P$ is the same as having a prefix in $P$.   By viewing $w$ as a word over the alphabet of $k$-generators, and looking at the  minimal prefix of this word whose image under $\alpha$ is  in $P$, we see
    \begin{align*}
\alpha^{-1}(P) \quad = \quad \alpha^{-1}(J) \ \cap\ \left( 
        \bigcup_{b,c}   \alpha^{-1}(b)\cdot \text{($k$-generators with value $c$)} \cdot \pw_k \right),
        \end{align*}
        where the sum in the parentheses ranges over choices of $b,c \in A$ such that $b$ is a strict infix of $J$ and $bc \in P$. In other words, the star-free language in the parentheses coincides with  the  prefix class of $a$ over contexts from $\alpha^{-1}(J)$. 
    Using a symmetric result for suffix classes, we see that there is a star-free language that coincides with the suffix class of $a$ over contexts with image in $J$. Intersecting these two, we get a language that coincides with $a$ over contexts with image in $J$, since $a$ is the unique element of its $\Hh$-class thanks to the assumption that all $\Hh$-classes in $J$ are trivial.
 \end{proof}

The above lemma, together with the following one, implies that all  elements of  $J$ are star-free.
 \begin{lemma}
     The language $\alpha^{-1}(J)$ is star-free.
 \end{lemma}
 \begin{proof}
    Let $L$ be the contexts whose image under $\alpha$ is neither in $J$ nor an infix of $J$.   If we take the complement of $L$, and then we remove the contexts with images that are strict infixes of $J$, which are star-free by induction assumption, then we are left with the language from the statement of the lemma. Therefore, it is enough to show that $L$ is star-free.   Using the same kind of argument as in Lemma~\ref{lem:approx-j}, which searches for a shortest prefix in $L$, we see that $L$ is defined by the expression
    \begin{align}\label{eq:bad-infix-expression}
        \bigcup_{b,c}   \alpha^{-1}(b)\cdot \text{($k$-generators with value $c$)} \cdot \pw_k,
        \end{align}
        where the sum ranges over choices of $b,c \in A$ such that $b$ is an infix (not necessarily strict) of $J$ and $bc$ is not an infix of $J$. The above expression is not yet known to be star-free, since it uses subexpressions of the form $\alpha^{-1}(b)$ for $b \in J$. However, such a subexpression can be replaced by star-free languages without affecting the value of the entire expression, as follows. 
   
        By Lemma~\ref{lem:approx-j} and the induction assumption, for every $b \in J$ there is a star-free language $L_b$ which coincides with $\alpha^{-1}(b)$ on contexts whose image is an infix of $J$. In other words, $L_b$  is equal to $\alpha^{-1}(b)$ plus some extra contexts which are in $L$. In the expression~\eqref{eq:bad-infix-expression}, replace each sub-expression $\alpha^{-1}(b)$ with $b \in J$ by the expression $L_a$. The extra contexts from the new sub-expressions are outside $L$, and therefore after this replacement the expression will still define $L$.  
 \end{proof}

This completes the proof of the induction step for infix classes where all elements are aperiodic.

\subsubsection{At least two bridges}
\label{sec:no-blockades}
We are left with the case when every context with image in $J$ has at least two bridges. This case is unique to graphs, and it will be resolved using the following lemma.   \begin{lemma}\label{lem:two-bridges-decomposition}
    If a context in $\pw_k$ has at least two bridges, then it can be decomposed as $w_1 \cdots w_\ell$ 
    so that  $\ell \le \Oo(k)$  and  each of the contexts $w_1,\ldots,w_\ell$ is either a $k$-generator, or has strictly more persistent ports than $w$.
\end{lemma}
Before proving the lemma, we use it to complete the proof of the induction step, thus completing the proof of  Theorem~\ref{thm:aperiodicity}. As we have observed in Claim~\ref{claim:number-of-persistent-ports-invariant}, the number of persistent ports is an invariant of the infix class $J$. Therefore, in the decomposition from Lemma~\ref{lem:two-bridges-decomposition}, each factor $w_i$ is either a $k$-generator, or its image  is a strict infix of $J$. Hence, we can apply Lemma~\ref{lem:reuse-bounded-infixes} to conclude that every element of  $J$ is star-free. 

It remains to prove Lemma~\ref{lem:two-bridges-decomposition}. One of the main ingredients of this proof will be a result based on~\cite{BODLAENDER1996358} and~\cite{bojanczykOptimizingTreeDecompositions2017a}, see Lemma~\ref{lem:typical-sequences} below,  which says that a path decomposition can be modified without affecting its width so that it does not alternate too much between different inner components.  To state this lemma, we need to  introduce some terminology for path decompositions. Recall that a  path decomposition is a sequence of sets of vertices, called bags. This sequence might contain repetitions, i.e.~the same set of vertices might appear  several  times in the sequence. To avoid ambiguity, use the word \emph{bag} to describe  an index in the sequence, and not  the set of vertices that is found at this index; and  hence different bags might contain the same vertices.  Recall that an {interval} in a path decomposition is  a set of bags which forms an interval with respect to the total ordering on bags.  We say that a vertex is \emph{active} in an interval if it is added or removed (or both) in the interval, i.e.~some bags from the interval have this vertex and some bags do not.

\begin{lemma}[Dealternation Lemma]\label{lem:typical-sequences}
    Let $w$ be a context  and let $X \cup Y$ be a partition of its non-port vertices such that there is no edge between $X$ and $Y$. If $w$ has a path decomposition of width $k$, then it also has a path decomposition of width $k$ which can be partitioned into $\Oo(k)$ intervals such that in each interval at most one of the sets $X$ or $Y$ is active. 
\end{lemma}
\begin{proof}
    This lemma is essentially proved in~\cite[Lemma 17]{bojanczykOptimizingTreeDecompositions2017a}, which itself is based on the typical sequences of Bodlaender and Kloks~\cite[p. 365]{BODLAENDER1996358}. The only purpose of this proof is   to introduce sufficient terminology so that we can apply the cited results.
    Fix some context $w$ of pathwidth $k$ for the rest of this proof; all path decompositions  in the proof will be path decompositions of this context.
    A path decomposition of $w$ can be viewed as a sequence of instructions from the set  
    \begin{align*}
    \set{\text{add($x$), remove($x$)} : \text{$x$ is a vertex of $w$}}
    \end{align*}
    such that:  (a) every vertex is added at most once and removed at most once in the sequence, and it cannot be added after  having been
     removed; (b) a vertex is added if and only if it is not a left port; (c) a vertex is removed if and only if it is not a  right port.  In particular, persistent ports are neither added nor removed.  An instruction sequence describes a path decomposition in the natural way. 
     
     To prove the lemma, we will begin with some path decomposition of width $k$, viewed as an instruction sequence in the sense described above, and then we will  modify it to achieve the conclusion of the lemma using a permutation of the instructions that respects the order within each of the sets $X$ and $Y$. Define a permutation of instructions in an instruction sequence to be \emph{separated} if the only instruction pairs whose mutual order is changed by the permutation are pairs where one instruction operates on $X$ and the other instruction operates on $Y$.  It is easy to see that if an instruction sequence is a path decomposition, then applying a separated permutation  yields a sequence of instructions that is also a path decomposition; here it is important that there are no edges connecting $X$ and $Y$. (The width might change, however.) We are now ready to apply~\cite[Lemma 17]{bojanczykOptimizingTreeDecompositions2017a}, which says that for every instruction sequence, one can apply a separated permutation so that in the new instruction sequence: (i) the width of the corresponding path decomposition does not increase; and (ii) the instructions can be grouped into $\Oo(k)$ intervals which operate either only on $X$, or only on $Y$, or only on ports. The resulting instruction sequences correspond to a path decomposition as required in the current lemma.
\end{proof}

Equipped with the Dealternation Lemma, we can  prove Lemma~\ref{lem:two-bridges-decomposition}.

\begin{proof}
    It is not hard to see that the lemma is equivalent to the following statement: (*) every context $w \in \pw_k$ with at least two bridges has a path decomposition of optimal width, whose bags can be partitioned into $\Oo(k)$ intervals so that in every interval there is some vertex that is present in all bags of the interval and which is not a persistent port of $w$.  The intervals in (*) correspond to the contexts $w_1,\ldots,w_\ell$, and the vertices present in all bags correspond to their new persistent ports. It remains to prove (*).

Suppose first that the context  $w$  has  an edge that goes directly from a non-persistent left port $x$ to a non-persistent right port $y$.  Every  path decomposition of $w$ -- including those of optimal width --   must have some   bag that contains both endpoints of this edge; all bags to the left of this bag contain $x$ and all bags to the right of this bag contain $y$. This proves (*), and therefore also the lemma, for contexts which have a bridge that is a single edge. 

Consider now a context that has at least two bridges, but which does not have any bridge that is a single edge as discussed in the previous paragraph. Take one of the bridges, and let $X$ be the non-port vertices that are incident to this bridge. Let $Y$ be the remaining non-port vertices. Since there are no edges connecting $X$ and $Y$, we can apply Lemma~\ref{lem:typical-sequences}, yielding an optimal width path decomposition and a partition of its bags into a family of at most $\Oo(k)$ intervals, call this family $\mathcal I$. By refining this family, we can assume without loss of generality that every port is inactive in every interval from $\mathcal I$. To prove (*), we need to show that for every interval $I \in \mathcal I$,  there is a  vertex that is not persistent port of $w$ and which appears in all bags of~$I$. If $I$ uses some non-persistent port vertex at least once, then this port vertex is present in all bags of $I$ and we are done. The interesting case is when $I$ does not use any port vertices except for the persistent ports.
  
        We know that all non-persistent left ports of $w$ are to the left of the interval $I$, and all non-persistent right ports of $w$ are to the right of the interval. Since one of the bridges requires using a vertex from $X$, and one of the bridges requires using a vertex from $Y$, it follows that the interval must use at least one vertex from $X$ and at least one vertex from $Y$. By  the Dealternation Lemma, either $X$ or $Y$ is inactive in the interval, and the inactive set will contribute a vertex that is present in all bags of the interval, thus proving (*).
\end{proof}

This completes the proof of the induction step in the proof that every element in $A$ is star-free, and thus also the proof of Theorem~\ref{thm:aperiodicity}.

\section{Future work}
We finish the paper with some potential directions for future work.
\begin{enumerate}
    \item {\bf Directed graphs.}  It is not clear how to generalize separator logic to directed graphs.
    \item {\bf Cliquewidth.} The star-free expressions used in this paper are based on operations  designed for treewidth. There are also operations designed for cliquewidth, and it is natural to ask about a variant of first-order logic that is equivalent to star-free expressions for these operations.
    \item {\bf Bounded treewidth.} It would be nice to generalize Theorem~\ref{thm:aperiodicity} from pathwidth to treewidth.  As mentioned previously, this would be hard even  for trees.
    \item {\bf Other algebraic characterizations for bounded pathwidth.} Over bounded pathwidth, one could attempt algebraic characterizations of other logics. One natural candidate is the usual variant of first-order logic with the edge relation only; in fact an algebraic characterization of this logic could even be attempted for bounded treewidth, as there corresponding logic for trees is already understood~\cite[Theorem 1]{benediktSegoufin2009}. Another natural candidate is  the  extension of first-order logic with predicates for disjoint paths from~\cite[Section 4]{schrader2021firstorder}.
\end{enumerate}
\bibliographystyle{plain}
\bibliography{bib}
\end{document}